\documentclass[12pt]{amsart}
\usepackage{verbatim}
\usepackage{amscd}

\numberwithin{equation}{section}
\input epsf

\begin{document}

\renewcommand{\theequation}{\thesection.\arabic{equation}}
\setcounter{secnumdepth}{2}
\newtheorem{theorem}{Theorem}[section]
\newtheorem{definition}[theorem]{Definition}
\newtheorem{lemma}[theorem]{Lemma}
\newtheorem{corollary}[theorem]{Corollary}
\newtheorem{proposition}[theorem]{Proposition}
\numberwithin{equation}{section}
\theoremstyle{definition}
\newtheorem{example}[theorem]{Example}
\title[Metric Reduction in Generalized Geometry and BTFT's]
{Metric Reduction in Generalized Geometry and Balanced Topological Field Theories}

\author[Yicao Wang]{Yicao Wang}
\address
{Department of Mathematics, Hohai University, Nanjing 210098, China}
\thanks{This study is supported by the Natural Science Foundation of Jiangsu Province (BK20150797).}
\maketitle

\baselineskip= 20pt
\begin{abstract}The recently established metric reduction in generalized geometry is encoded in 0-dimensional supersymmetric $\sigma$-models. This is an example of balanced topological field theories. To find the geometric content of such models, the reduction of Bismut connections is studies in detail. Generalized K$\ddot{a}$hler reduction is briefly revisited in this formalism and the generalized K$\ddot{a}$hler geometry on the moduli space of instantons on a generalized K$\ddot{a}$hler 4-manifold of even type is thus explained formally in a topological field theoretic way.
\end{abstract}

\section{Introduction}
In this paper, we will investigate the relationship between balanced topological field theories, the name of which was coined in \cite{DM} in 1997, and the reduction theory in the more recent generalized geometry; in particular, we will show how a large portion of the metric reduction theory developed in \cite{BCG1} \cite{Ca} fits into balanced topological field theories. This investigation sheds some new light on both balanced topological field theories and generalized geometry.

In generalized geometry, the notion of a generalized complex manifold is a simultaneous generalization of complex and symplectic manifolds. Since the famous (Marsden-Weinstein) symplectic reduction is an important construction in symplectic geometry, it is natural to establish a certain "generalized" reduction theory. This was solved in great generality in \cite{BCG1}. However, there are still some subtleties in the new reduction theory which should be clarified. We just mention one of these subtleties: In the classical symplectic reduction, a moment map takes its values in the dual space $\mathfrak{g}^*$ of the Lie algebra $\mathfrak{g}$ acting on the manifold $M$, while the one in \cite{BCG1} can take values in a more general $\mathfrak{g}$-module. In the balanced field theoretic approach to the reduction theory presented in this article, to some extent, we explain this novelty in terms of more traditional ideas: A moment map in \cite{BCG1} is interpreted here as a section of a trivial equivariant vector bundle. However, we use more general equivariant bundles and this naturally generalizes the notion of moment maps in \cite{BCG1}.

Let us explain several hints leading us to relate balanced topological field theories with metric reduction in generalized geometry. The first hint has its origin in 2-dimensional supersymmetric $\sigma$-models, which are also the main motivating sources of generalized geometry. The most general $N=(1,1)$ action, defined on a 2-dimensional Minkowski space $\Sigma$, is of the following form:
\begin{equation}S(\varphi)=\frac{1}{2}\int d^2 \sigma d\theta^+d\theta^-E_{ij}(\varphi)D_+\varphi^iD_-\varphi^j,\label{2daction}\end{equation}
where $E_{ij}=g_{ij}+B_{ij}$ for a Riemannian metric $g$ and a 2-form $B_{ij}$ over the target space, and $D_\pm=\partial_{\theta^\pm}+\sqrt{-1}\theta^\pm\partial_{\pm}$. The two supercharges are $Q_\pm=\partial_{\theta^\pm}-\sqrt{-1}\theta^\pm\partial_{\pm}$,
satisfying the following anti-commutative relations (part of the $N=(1,1)$ supersymmetry algebra):
\[\{Q_\pm, Q_\mp\}=0,\quad \{Q_\pm, Q_\pm\}=P_\pm,\]
where $P_\pm$ are the generators of infinitesimal translation in the Poincare group in 2 dimensions. If the model is reduced to zero dimension by dimensional reduction, the supersymmetry algebra turns out to be
\[\{Q_m, Q_n\}=0,\quad m, n=\pm.\]
\emph{This is precisely the anti-commutative relations satisfied by topological charges in a balanced field theory.}

The second hint is that, it is well-known that the mathematical reduction theory can be physically realized by gauging a $\sigma$-model carrying global symmetries. A large part of ordinary reduction theory can then be encoded in an N=1 topological gauge theory, which computes the Euler number of a certain vector bundle over some moduli space. These ideas were developed into the so-called Mathai-Quillen formalism of N=1 topological field theories \cite{Wu}. Though it seems that the generalized reduction theory cannot directly fit into the Mathai-Quillen formalism, it is reasonable to conjecture that there should be some analogous formalism to encode the generalized theory.

The third hint is more concrete. In view of the reduction theory of \cite{BCG1}, the reduction procedure involves two basic steps: the Courant reduction to an invariant submanifold $N\subset M$ and then the Courant reduction to the quotient space $N/G$, where $G$ is the group acting on $M$. Actually, these two steps are precisely what Blau and Thompson considered in \cite{BT1} in the context of N=2 topological gauge theories, which was later revisited from the angle of balanced topological field theories in \cite{DM}.

Our approach is not simply an application of balanced topological field theories, but contains some new points which were not covered before. In \cite{DM} when writing down a general action for a balanced topological field theory, the authors insisted that it should be $sl(2)$-invariant (see Sect.~\ref{MI}). However, in the presence of the NS-flux $H$ which is essential in generalized geometry, this invariance is broken and new geometry arises. In this sense, our investigation extends the content of \cite{DM}. Motivated by generalized geometry, we also suggest some possible future directions in balanced theories in the last section.

This paper is motivated directly by the investigation in \cite{BCG2}. N. Hitchin discovered in \cite{Hi2} that the moduli space of instantons over a generalized K$\ddot{a}$hler 4-manifold of even type is equipped with a natural generalized K$\ddot{a}$hler structure, and he asked whether this generalized K$\ddot{a}$hler structure could be viewed as obtained from certain "generalized K$\ddot{a}$hler reduction" procedure, just as in the genuine K$\ddot{a}$hler case. This question was affirmatively answered in \cite{BCG2} by applying the reduction theory developed in \cite{BCG1} to this infinite dimensional case formally. The moduli space of instantons is very important in topological field theories. To understand the work of \cite{BCG2} in some depth, we could further ask ourselves what the underlying topological field theoretical content of \cite{BCG2} is. The present work partly arises as an attempt to look for an answer to this question.

The outline of the present article is as follows. In Sect.~\ref{MI}, we introduce our Model I, a zero-dimensional supersymmetric $\sigma$-model. This is a balanced topological field theory whose partition function computes the Euler characteristic of the target space. In Sect.~\ref{MII}, we give our Model II, which is obtained by gauging Model I when there is a global symmetry group $G$. To compare with computations in generalized geometry, we carry out the localization of the path-integral explicitly. The reduced model is then again our Model I with the quotient space as the target space. In Sect.~\ref{Model3}, we extend our Model I to Model III which involves an auxiliary vector bundle $W$ with a generic section $\sigma$. The localization then gives rise to Model I with the zero-locus $\sigma^{-1}(0)$ as the target space. In Sect.~\ref{M4}, we combine the former constructions together to give the most general Model IV. Due to the detailed analysis of the basic constructions in former section, the localization is only sketched briefly. The reduced model is of course our Model I with the quotient $\sigma^{-1}(0)/G$ as the target space. To see the several reduced models really compute the Euler characteristic of certain manifolds, in Sect.~\ref{Bismut}, we derive the curvatures of the $-$-Bismut connections in the reduced models, with a quotient or a submanifold as the target space. This computation is possible due to the observation in \cite{Hi} \cite{Gu1} that the Bismut connections can be expressed using Courant brackets and the fact that the reduced Courant bracket can be expressed in terms of Courant reduction in the sense of \cite{BCG1}. It is showed that the purely geometric computations do coincide precisely with the physical interpretation. With the metric reduction in place, we briefly discuss generalized K$\ddot{a}$hler reduction in Sect.~\ref{GKR}. In Sect.~\ref{MOI}, as an application of our approach to the reduction theory, we revisit the work in \cite{BCG2} and clarify its underground field theoretic content. Since the reduction procedure in this setting was mathematically analyzed in detail in the literature, our main new contribution is to write down the action of the underlying topological field theory. The last section outlines some future problems motivated by our present investigation.

Let us add another comment on our approach to generalized geometry. Physically, there are two ways to obtain an $N=(2,2)$ supersymmetric $\sigma$-model: On one side, one can start from an $N=(1,1)$ model and extend it to on-shell $N=(2,2)$ supersymmetry by introducing further geometric structures on the target space. On the other side, one can also try to construct off-shell $N=(2,2)$ models directly from $N=(2,2)$ superspace techniques. The latter approach involves a further complicated classification of superfields into chiral-, twisted chiral- and even semichiral-superfields. The off-shell formulation of $N=(2,2)$ $\sigma$-models was only resolved recently in \cite{Lin}. However, in the most general case, there are singular points on the target space around which the above classification breaks down. Our treatment of the zero-dimensional analogue of 2-dimensional models goes in the spirit of the $N=(1,1)$ approach, avoiding the possible singularities arising in the $N=(2,2)$ approach.

There are also other ways to approach the generalized reduction theory by gauging supersymmetric $\sigma$-models. In \cite{Me} the reduction of generalized K$\ddot{a}$hler structures was considered by gauging 2-dimensional off-shell $N=(2,2)$ supersymmetric $\sigma$-models. Later in \cite{Ka} the same topic was revisited by gauging 2-dimensional $N=(1,1)$ supersymmetric $\sigma$-models. However, the cases investigated in these papers were a special one called Hamiltonian reduction, the counterpart of Marsden-Weinstein reduction in symplectic geometry. This case was mathematically investigated in detail in \cite{LT0}.
\section{Model I}\label{MI}
As noted in \cite{DM}, an N=1 theory treats the geometry of the supermanifold $\Pi TM$, i.e. the tangent bundle of a smooth manifold $M$ with the parity of the fiber being reversed, while an N=2 theory treats the geometry of the \emph{iterated superspace} $\Pi T(\Pi TM)$. Despite this similarity, compared with N=1 topological field theories, N=2 topological ones are not well-developed. We refer the reader to \cite{DM} \cite{BT2} for the basics of N=2 topological field theories.

Our starting point is a Riemannian manifold $(M, g)$ together with a NS-flux $H$, i.e. a closed 3-form. The triple $(M, g, H)$ is also called a generalized Riemannian manifold. There are two topological charges $d_\pm$ in a balanced topological theory, with the following anti-commutative relations:
\[\{d_m, d_n\}=0,\quad m, n=\pm.\]
The N=2 scalar superfields are of the form \[\varphi^i=x^i+\theta^+\psi_+^i+\theta^-\psi_-^i+\theta^+\theta^-\tilde{F}^i,\]
where $\tilde{F}$ lives in the 2-jet bundle of $M$.
The components of $\varphi$ obey
\[d_m x^i=\psi_m^i,\quad d_m \psi_n^i=\epsilon_{mn}\tilde{F}^i,\quad d_m \tilde{F}^i=0.\]
Note that all throughout the paper, we follow the convention that $-\epsilon_{+-}=\epsilon^{+-}=1$.

To write the action in a covariant way, we shall equip $M$ with the Levi-Civita connection and introduce the auxiliary fileds
\[F^i=\tilde{F}^i+(\Gamma_-\psi_+)^i=\tilde{F}^i+\Gamma_{jk}^i\psi_-^j\psi_+^k,\]
where $\Gamma_jdx^j$ is the connection form of the Levi-Civita connection. Let $R$ be the curvature of this connection.

There is a natural $sl_2$-action on the field content, generated by three operators
\[L_{\pm\pm}=\psi_\pm^i\frac{\partial}{\partial \psi_\mp^i},\quad L_{+-}=\psi_+^i\frac{\partial}{\partial \psi_+^i}-\psi_-^i\frac{\partial}{\partial \psi_-^i}.\]
Note that $L_{+-}$ computes the ghost number. The full $sl_2$-invariance would prevent the potential $B$ from appearing in the following action (\ref{0action}). However, the NS-flux $H=dB$ even features generalized geometry. So in our case, we only keep the requirement that the action be $L_{+-}$-invariant, namely the ghost number of the action should be zero.

The first model we shall consider in this paper is a model without any non-trivial group action and with no extra vector bundle $W$. We call it Model I, the action of which is of the following form:\footnote{We follow the convention that $\int d\theta^+ d\theta^- \theta^-\theta^+=1$.}
\begin{equation}S(x, \psi_\pm ,F)=\frac{1}{2}\int d\theta^+d\theta^- E_{ij}(\varphi)Q_-\varphi^i Q_+\varphi^j,\label{0action}\end{equation}
where $E_{ij}=g_{ij}-B_{\alpha ij}$, $B_\alpha$ is a local potential of $H$ over a coordinate patch $U_\alpha$ of $M$, and $Q_\pm=\partial_{\theta^\pm}$. In terms of components,
\begin{equation}S_I=-\frac{1}{2}(\psi_-, R^-\psi_-)+\frac{1}{2}(F+\frac{1}{2}H_{ij}^{\quad .}\psi_-^i\psi_+^j, F+\frac{1}{2}H_{{i'}{j'}}^{\quad \ast}\psi_-^{i'}\psi_+^{j'}),\end{equation}
where $R^-$ is the curvature of the $-$-Bismut connection $\nabla^-=\nabla-\frac{1}{2}g^{-1}H$, i.e.
\[R^-_{ijkl}=R_{ijkl}-\frac{1}{2}(\nabla_iH_{jkl}-\nabla_jH_{ikl})+\frac{1}{4}(H_{ipl}H_{jk}^{\quad p}-H_{jpl}H_{ik}^{\quad p}).\]
Note that $S$ is well-defined globally, though $B_\alpha$ need not be. This action can also be derived from the action (\ref{2daction}) by dimensional reduction.

Integrating out the auxiliary field $F$, we get the partition function of the model\footnote{In the case $H=0$, the $sl_2$-invariance follows from the Bianchi identity $R_{ijkl}+R_{jkil}+R_{kijl}=0$ which fails when $H\neq 0$.}:
\[Z=\int \frac{dxd\psi_+d\psi_-}{\sqrt{g}} e^{\frac{1}{2}(\psi_-, R^-\psi_-)},\]
which, up to a multiplicative constant, is the Euler characteristic of $M$ as expected. The result can equally be expressed in terms of the curvature $R^+$ of $\nabla^+=\nabla+\frac{1}{2}g^{-1}H$ since $R^-_{ijkl}=R^+_{klij}$. For latter use, we also write down the following on-shell supersymmetric transform\footnote{Namely,  $\tilde{F}^i$ in an off-shell expression is replaced with $\Gamma^{(-)i}_{jk}\psi_-^j\psi_+^k$ or $\Gamma^{(+)i}_{kj}\psi_-^j\psi_+^k$ .}
\begin{eqnarray}d_+\psi_-^i=(\Gamma_{jk}^{\quad i}+\frac{1}{2}H_{jk}^{\quad i})\psi_-^j\psi_+^k=\Gamma^{(-)i}_{kj}\psi_-^j\psi_+^k,\label{onshell}\end{eqnarray}
where $\Gamma^{(-)}_jdx^j$ is the connection form of $\nabla^-$. This is how the Bismut connection arises in our model.

The topological supersymmetry can be extended further to N=4 on-shell supersymmetry if $M$ is equipped with a generalized K$\ddot{a}$hler structure. Recall from \cite{Gua1} that a generalized K$\ddot{a}$hler structure has an equivalent biHermitian description: There are two almost complex structures $J_\pm$ compatible with $g$, satisfying $\nabla^\pm J_\pm=0$ and that $H$ should be of type $(2,1)+(1,2)$ w.r.t. both of $J_\pm$, where $\nabla^\pm=\nabla\pm\frac{1}{2}g^{-1}H$. If $d_\pm'$ is the second pair of differentials besides $d_\pm$, the N=4 algebra is
\[\{d_m, d_n\}=\{d_m', d_n'\}=\{d_m, d_n'\}=0, \quad m,n=\pm.\]

With a generalized K$\ddot{a}$hler structure in place, the extended supersymmetric transform is of the form:
\[\delta'_\epsilon \varphi^i=\epsilon^+ J_{+j}^i(\varphi)Q_+\varphi^j+\epsilon^-J_{-i}^i(\varphi)Q_-\varphi^j.\]
In components, the extended transform is
\[d_\pm'x^i=J_{\pm j}^i\psi_\pm^j,\quad d'_\pm \tilde{F}^i=J^i_{\pm j,k}(\tilde{F}^k\psi_\pm^j-\tilde{F}^j\psi_\pm^k),\]
\[d'_\pm\psi_\pm^i=-J^i_{\pm j,k}\psi_\pm^k\psi_\pm^j,\quad d'_\pm\psi_\mp^i=\mp J^i_{\pm j}\tilde{F}^j-J^i_{\pm j,k}\psi_\mp^k\psi_\pm^j.\]
The on-shell form of these formulae is
\[d_\pm'x^i=J_{\pm j}^i\psi_\pm^j,\]
\[d'_\pm\psi_\pm^i=[\Gamma^{(\pm)i}_{kl}J_{\pm j}^l-\Gamma^{(\pm)l}_{kj}J^i_{\pm l}]\psi_\pm^k\psi_\pm^j,\]
\[d'_\pm\psi_\mp^i=\Gamma_{kl}^{(\pm)i}J_{\pm j}^l\psi_\mp^k\psi_\pm^j.\]

It should be remarked that, unlike the 2-dimensional case where $M$ being generalized K$\ddot{a}$hlerian is both sufficient and necessary for the model to acquire on-shell $N=(2,2)$ supersymmetry, in the zero-dimensional case, this is only sufficient for Model I to acquire on-shell N=4 supersymmetry. Besides, one can also consider off-shell N=3 supersymmetry in Model I by dropping either of $J_\pm$, and geometrically this can be used to investigate SKT geometry.

\section{Model II}\label{MII}
\subsection{The gauged model}
In this subsection we shall gauge Model I to encode the action of a Lie group $G$ of dimension $s$. As for the basic underlying algebraic and geometric structures of this equivariant setting, we refer the reader to  \cite{BT2} \cite{DM}. In this context, new fields $\phi_{mn}^a$, $\eta_m^a$ ($\phi_{mn}^a=\phi_{nm}^a, m,n=\pm, a=1,2,\dots, s$) carrying extra group indices should be introduced. $\phi$ and $\eta$ are even and odd respectively. They form the superfields
\[A_+^a=\theta^+ \phi_{++}^a+\theta^-\phi_{+-}^a+2\theta^+\theta^-\eta_+^a,\]
and
\[A_-^a=\theta^+ \phi_{+-}^a+\theta^-\phi_{--}^a+2\theta^+\theta^-\eta_-^a.\]
This corresponds to the Cartan model of N=2 equivariant theory or Wess-Zumino gauge in physical terms.

The N=2 equivariant field content fulfills
\[d_n\phi_{mp}=\epsilon_{nm}\eta_p+\epsilon_{np}\eta_m,\quad d_n\eta_m=\frac{1}{2}\epsilon^{pq}[\phi_{nq},\phi_{mp}].\]
\[d_nx^i=\psi_n^i,\quad d_n \psi_m^i=\phi_{nm}^a V_a^i+\epsilon_{nm}\tilde{F}^i,\]
\[d_n \tilde{F}^i=-\mathcal{L}(\phi_{nm})\psi_p^i\epsilon^{mp}-\mathcal{L}(\eta_n)x^i,\]
where $\mathcal{L}$ denotes Lie derivative, for example,
\[\mathcal{L}(\phi_{nm})\psi_p^i=\phi_{nm}^a\frac{\partial V_a^i}{\partial x^j}\psi_p^j,\]
and
\[\mathcal{L}(\eta_n)x^i=\eta_n^a V_a^i,\]
where $V_a^i\partial_{x^i}$ are the fundamental vector fields generated from a basis $\{e_a\}$ of the Lie algebra $\mathfrak{g}$ of $G$.

In the gauged model and under the Wess-Zumino gauge, $Q_\pm$ is replaced by $\mathcal{Q}_\pm$:
\[\mathcal{Q}_\pm\varphi^i=\partial_{\theta^\pm}\varphi^i+A_\pm^aV_a^i(\varphi).\]
The action in terms of superfields is
\begin{equation}S(\varphi, A)=\frac{1}{2}\int d\theta^+d\theta^-E_{ij}(\varphi)\mathcal{Q}_-\varphi^i\mathcal{Q}_+\varphi^j.\label{actionII}\end{equation}

Now to have a globally well-defined action, the local potentials $B_\alpha$ should meet some further requirements. First $B$ should be $G$-invariant (therefore $H$ is invariant). Besides, motivated by observations in \cite{BCG1}, we define $\xi_a=-\iota_a B_\alpha$ and require it to be a global equivariant 1-form on $M$. Then $H+\phi^a \xi_a$ is an closed equivariant 3-form and $\{V_a+\xi_a\}$ form an \emph{isotropic trivially extended action} of $\mathfrak{g}$ in terminology of \cite{BCG1}. Conversely, given a closed equivariant extension $H+\phi^a \xi_a$ of $H$ such that $\{V_a+\xi_a\}$ form an isotropic trivially extended $\mathfrak{g}$-action, we can find local potentials $B$ for both $H$ and $\xi_a$, at least when $G$ is compact and connected and the $G$-action is free.
\begin{theorem} Let compact connected Lie group $G$ act freely on $M$. If $H+\phi^a\xi_a$ is a closed equivariant extension of the $G$-invariant 3-form $H$ such that $\xi_a(V_b)+\xi_b(V_a)=0$, then there exist local $G$-invariant 2-forms $B$ such that
\begin{equation}d_\phi B=H+\phi^a\xi_a,\label{poten}\end{equation}
where $d_\phi=d-\phi^a\iota_a$ is the equivariant de Rham differential. Such $B$ are unique up to basic closed 2-form.
\end{theorem}
The proof of this theorem can be found in \cite{Wang} and due to this result, a large part of the cases considered in \cite{BCG1} \cite{Ca} can really be encoded in our model.

In components, the action (\ref{actionII}) is
\begin{eqnarray*}
S_{II}&=&S_I-\frac{1}{2}G_{ab}\phi_{+-}^a\phi_{+-}^b+\frac{1}{2}(\nabla_+\xi_{ai}\psi_-^i-\nabla_-\xi_{ai}\psi_+^i)\phi_{+-}^a-\xi_{ai}F^i\phi_{+-}^a\\
&-&(\psi_+, \mu_a\psi_-)\phi_{+-}^a+\frac{1}{2}G_{ab}\phi_{++}^a\phi^b_{--}+\frac{1}{2}\xi_{ai}V_b^i\phi_{--}^a\phi_{++}^b+\frac{1}{2}\psi_-^i\nabla_-\xi_{ai}\phi_{++}^a
\\&-&\frac{1}{2}\psi_+^i\nabla_+\xi_{ai}\phi_{--}^a-\frac{1}{2}(\mu_a\psi_-, \psi_-)\phi_{++}^a-\frac{1}{2}(\mu_a\psi_+, \psi_+)\phi_{--}^a\\
&+&\eta_+^a[(V_a, \psi_-)-\xi_{ai}\psi_-^i]-\eta_-^a[(V_a, \psi_+)+\xi_{ai}\psi_+^i],
\end{eqnarray*}
where $G_{ab}=(V_a, V_b)$ and $\mu_a \psi_\pm:=-\nabla_{\pm}V_a$ is the Riemannian moment map of the $\mathfrak{g}$-action.
\subsection{Localization}
\label{LII}
In this subsection, we assume that $G$ acts freely and properly on $M$. We shall compute the path-integral $\int dxd\psi_\pm dF d\phi d\eta e^{-S_{II}}$ explicitly. The result shows that the integral over the zero modes is actually the Euler characteristic of the quotient space $M/G$. However, this will be clear only after we would have derived the curvature of the Bismut connection over $M/G$ in Sect.~\ref{Bcur}.
\begin{itemize}
\item Denote $V_a^{\pm}=V_a\pm g^{-1}\xi_a$. The effect of integrating out $\eta_\pm$ is to restrict $\psi_\pm$ to the zero modes:
\begin{equation}(V_a^\pm, \psi_\pm)=0, \quad a=1,\cdots, s.\label{con}\end{equation}
This determines two horizontal distributions $\tau_\pm$ in $M$ viewed as a principal $G$-bundle over $M/G$. We will see later in Sect.~\ref{Bismut} that these also arise naturally in metric reduction in generalized geometry.
\end{itemize}
\begin{itemize}
\item Integrate out $\phi_{++}$. The relevant terms are
\[\frac{1}{2}K_{ab}\phi_{++}^a\phi_{--}^b+\frac{1}{2}\phi_{++}^a(\nabla_-V^-_a, \psi_-),\]
where $K_{ab}=G_{ab}-\xi_{ai}V_b^i=E(V_b,V_a)$. This will give rise to a Dirac $\delta$-function, restricting $\phi_{--}$ to\footnote{Actually, to produce the $\delta$-function, a certain factor $\sqrt{-1}$ should be included consistently, but we won't do this in detail here.}
\begin{equation}\phi_{--}^a=-K^{ab}(\nabla_-V^-_b, \psi_-),\label{phi}\end{equation}
where $K^{ab}$ is the inverse of $K_{ab}$, i.e. $K^{ab}K_{bc}=\delta^a_c$.
\end{itemize}
\begin{itemize}
\item Integrate out $\phi_{--}$. This is equivalent to replacing $\phi_{--}$ in the action by the R.H.S. of Eq.~(\ref{phi}). The relevant terms produce an exponent
    \[-\frac{1}{2}K^{ab}(\nabla_-V^-_b, \psi_-)(\nabla_+V_a^+, \psi_+).\]
\end{itemize}
\begin{itemize}
\item Integrate out $F$ or instead substitute the equation of motion of $F$
\begin{equation}F^k+\frac{1}{2}H_{ij}^{\quad k}\psi_-^i\psi_+^j-\phi_{+-}^a \xi_{ai}g^{ik}=0 \label{EF}\end{equation}
into the following expression:
\[\frac{1}{2}(F+\frac{1}{2}H_{ij}^{\quad .}\psi_-^i\psi_+^j, F+\frac{1}{2}H_{{i'}{j'}}^{\quad \ast}\psi_-^{i'}\psi_+^{j'})-\xi_{ai}F^i\phi_{+-}^a.\]
We get
\[-\frac{1}{2}g^{ij}\xi_{ai}\xi_{bj}\phi_{+-}^a\phi_{+-}^b+\frac{1}{2}H_{jk}^{\quad i}\xi_{ai}\psi_-^j\psi_+^k\phi_{+-}^a.\]
\end{itemize}
\begin{itemize}
\item Integrate out $\phi_{+-}$. The relevant terms are
\[-\frac{1}{2}T_{ab}\phi_{+-}^a\phi_{+-}^b-\frac{1}{2}\phi_{+-}^a[(\nabla_+V_{a}^-,\psi_-)+(\nabla_-V_{a}^+,\psi_+)-H_{ij}^{\quad k}\psi_-^i\psi_+^j\xi_{ak})],\]
where, due to $\xi_a(V_b)=-\xi_b(V_a)$, \[T_{ab}=(V_a^+, V_b^+)=(V_a^-,V_b^-)=G_{ab}+g^{ij}\xi_{ai}\xi_{bj}.\]
These give an exponent of the form
\begin{eqnarray*}&\frac{1}{8}&T^{ab}[(\nabla_+V_{a}^-,\psi_-)+(\nabla_-V_{a}^+,\psi_+)-H_{ij}^{\quad k}\psi_-^i\psi_+^j\xi_{ak})]\\&\times&[(\nabla_+V_{b}^-,\psi_-)+(\nabla_-V_{b}^+,\psi_+)-H_{i'j'}^{\quad l}\psi_-^{i'}\psi_+^{j'}\xi_{bl}]\end{eqnarray*}
where $T^{ab}$ is the inverse of the matrix $T_{ab}$.
\end{itemize}
Combining the above calculations together, we finally get the reduced action involving only the bosonic fields $x^i$ and zero modes of $\psi_\pm$:
\begin{eqnarray}
S_{IR}&=&-\frac{1}{2}(\psi_-, R^- \psi_-)-\frac{1}{2}K^{ab}(\nabla_+V_a^+, \psi_+)(\nabla_-V^-_b, \psi_-)\nonumber\\&+&\frac{1}{8}T^{ab}[(\nabla_+V_{a}^-,\psi_-)+(\nabla_-V_{a}^+,\psi_+)-H_{ij}^{\quad k}\psi_-^i\psi_+^j\xi_{ak})]\nonumber\\&\times&[(\nabla_+V_{b}^-,\psi_-)+(\nabla_-V_{b}^+,\psi_+)-H_{i'j'}^{\quad l}\psi_-^{i'}\psi_+^{j'}\xi_{bl}].\label{CurII}
\end{eqnarray}
This is actually Model I with $M/G$ as the target space. We will show the above expression is precisely the curvature of the Bismut connection on $M/G$.

To find the reduced NS-flux $\tilde{H}$ on $M/G$, we note that, using Eq.~(\ref{EF}), it is obtained that
\begin{equation}d_+\psi_-^i=\phi_{+-}^aV_a^{-i}+\Gamma^{(-)i}_{kj}\psi_-^j\psi_+^k,\label{bis}\end{equation}
where \[\phi_{+-}^a=-\frac{1}{2}T^{ab}[(\nabla_+V_{a}^-,\psi_-)+(\nabla_-V_{a}^+,\psi_+)-H_{ij}^{\quad k}\psi_-^i\psi_+^j\xi_{ak})]\]
and the constraint (\ref{con}) is imposed on $\psi_\pm$. Compared with Eq.~(\ref{onshell}), Eq.~(\ref{bis}) gives the $-$-Bismut connection $\tilde{\nabla}$ over $M/G$ in terms of the connection $\nabla^-$ over $M$ and the two distributions $\tau_\pm$. Then $-\tilde{H}$ is the torsion of $\tilde{\nabla}$. The detailed computation will be carried out in Sect.~\ref{Bismut}.

Before finishing this section, we comment that we have only paid attention to the exponent in the reduced path-integral and ignored the several factors arising from the computation and even haven't fixed the gauge.
\section{Model III}\label{Model3}
In this section, we consider a model with an extra vector bundle $W$ of rank $r$ (together with a generic section $\sigma\in \Gamma(W)$) but without any non-trivial group action. This can be viewed as Model I with the target space being the total space of $W^*$, the dual of $W$.

New superfields $\zeta_\alpha$ ($\alpha=1,2,\cdots, r$) living in $W^*$ should be introduced:
\[\zeta=U+\theta^+\tilde{\chi}_++\theta^-\tilde{\chi}_-+\theta^+\theta^-\tilde{L}.\]
However, to write the action in a covariant way, we have to choose two connections $\nabla^\pm$ in $W$ and introduce new fields $\chi_\pm$, $L$. These are defined as
\[\chi_\pm:=\nabla^\pm U:= d_\pm U+\omega^\pm U=\tilde{\chi}_\pm+\omega^\pm U,\]
and
\[ \quad L=\frac{1}{2}(\nabla^-\chi_+-\nabla^+\chi_-),\]
where $\omega^\pm=\omega^\pm_i\psi_\pm^i$ are the connection forms of $\nabla^\pm$. $L$ and $\tilde{L}$ are related by a more complicated formula, but we won't use it explicitly and so omit it here.

Besides the term of Model I, the action includes a new term involving $\sigma$:
\[S_\sigma(\varphi, \zeta)=\int d\theta^+d\theta^- \sqrt{-1}(\sigma(\varphi), \zeta).\]
We call this Model III. Note that as in \cite{DM}, another term of the form $d_+d_-(\chi_+, \chi_-)$ could be added if $W$ is equipped with a metric. We won't introduce such a term and consequently the superfield $\zeta$ serves only as a Lagrangian multiplier in our model. This is similar to what Blau and Thompson did in \cite{BT1}, where N=2 topological gauge theories were treated in a supersymmetric quantum mechanical formalism.

In components, the action of Model III is
\begin{eqnarray*}S_{III}&=&S_I+\sqrt{-1}[(\sigma,L)+(\nabla^-\sigma, \chi_+)-(\nabla^+\sigma, \chi_-)
\\&-&\frac{1}{2}(\{\nabla_j^+, \nabla_i^-\}\sigma, U)\psi_+^j\psi_-^i+(\overline{\nabla}_i \sigma, U)\tilde{F}^i],\end{eqnarray*}
where $\overline{\nabla}_i=(\nabla^+_i+\nabla^-_i)/2$.

In the following, we consider the localization of this model.
\begin{itemize}
\item Integrating out $L$, we get a Dirac $\delta$-function restricting the bosonic fields $x^i$ to lie in $N:=\sigma^{-1}(0)$.
\end{itemize}
\begin{itemize}
\item Integrate out $\chi_\pm$. Since $\sigma$ is generic, only zero-modes (tangent to $N$) of $\psi_\pm$ are left. After this computation, the terms left in $S_{\sigma}$  are
\begin{equation*}-\sqrt{-1}(\partial_i\partial_j\sigma, U)\psi_+^j\psi_-^i+\sqrt{-1}(\partial_i \sigma, U)(F^i-(\Gamma_-\psi_+)^i).\end{equation*}
Note that the zero modes of $F$ won't contribute to the second term.
\end{itemize}
\begin{itemize}
\item Integrate out the transverse modes of $F$ in the normal directions of $N\subset M$. The equation of motion of such modes of $F$ is
\[F^k+\frac{1}{2}H_{ij}^{\quad k}\psi_-^i\psi_+^j+\sqrt{-1}(\partial_l\sigma, U)g^{lk}=0.\]
The integration gives an exponent
\[\frac{1}{2}T^{\alpha\beta}U_\alpha U_\beta-\frac{\sqrt{-1}}{2}\partial_i\sigma^\alpha U_\alpha H_{kl}^{\quad i}\psi_-^k\psi_+^l,\]
where $T^{\alpha\beta}=\partial_i \sigma^\alpha\partial_j \sigma^\beta g^{ij}$.
\item Integrate out the zero modes of $F$. Since they only occur in a complete square and hence have no contribution at all, they can be simply ignored.
 \item Integrate out $U$. The relevant terms are
 \[\frac{1}{2}T^{\alpha\beta}U_\alpha U_\beta+\sqrt{-1}U_\alpha\nabla_i^+\partial_j\sigma^\alpha\psi_-^i\psi_+^j,\]
 which give rise to the exponent
\begin{equation}\frac{1}{2}T_{\alpha\beta}\nabla_i^+\partial_j\sigma^\alpha\nabla_{i'}^+\partial_{j'}\sigma^\beta \psi_-^i\psi_+^j\psi_-^{i'}\psi_+^{j'},\label{CurIII}\end{equation}
where $T_{\alpha\beta}$ is the inverse of $T^{\alpha\beta}$. Note that here $\nabla^+$ is the Bismut connection $\nabla+\frac{1}{2}g^{-1}H$, rather than the connection $\nabla^+$ in $W$. The above result together with $-(\psi_-, R^- \psi_-)/2$ turns out to be the curvature of the $-$-Bismut connection on $N$. Thus, we again get Model I with the submanifold $N$ being the target space. This will be clear in Sect.~\ref{Bismut}.
\end{itemize}

Now it is easy to find the following on-shell supersymmetric transform
\begin{equation}d_+\psi_-^i=(\Gamma^{(-)i}_{kj}+T_{\alpha\beta}g^{il}\partial_l\sigma^\alpha\nabla^+_j \partial_k \sigma^\beta)\psi_-^j\psi_+^k,\label{sbis1}\end{equation}
where $x^i$ are restricted on $\sigma^{-1}(0)$ and $\psi_\pm$ on the tangent space of $\sigma^{-1}(0)$.

\section{Model IV}\label{M4}
In this section, we shall combine all ingredients discussed before together. The new model should involve the $G$-action and an extra equivariant vector bundle $W$ together with a generic equivariant section. Now the action takes the following form:
\begin{equation}S_{IV}(\varphi, \zeta)=\int d\theta^+d\theta^-[\sqrt{-1}(\sigma, \zeta)+\frac{1}{2}E_{ij}(\varphi)\mathcal{Q}_-\varphi^i\mathcal{Q}_+\varphi^j],\label{IV1}\end{equation}
In terms of components, the full action is
\begin{eqnarray}
S_{IV}&=&S_{III}-\frac{\sqrt{-1}}{2}\phi_{+-}^a V_a^i((\nabla_{-i}-\nabla_{+i})\sigma, U)+\frac{1}{2}\phi_{+-}^a(\nabla_+\xi_{ai}\psi_-^i-\nabla_-\xi_{ai}\psi_+^i)\nonumber\\&-&\phi_{+-}^a\xi_{ai}F^i
-\frac{1}{2}G_{ab}\phi_{+-}^a\phi_{+-}^b+\frac{1}{2}\xi_{ai}V_b^i\phi_{--}^a\phi_{++}^b+\frac{1}{2}G_{ab}\phi_{++}^a\phi_{--}^b\nonumber\\
&-& \phi_{+-}^a (\psi_+, \mu_a \psi_-)+\frac{1}{2}\phi_{++}^a(\nabla_-V_a^-, \psi_-)+\frac{1}{2}\phi_{--}^a(\nabla_+V_a^+, \psi_+)\nonumber\\&+&\eta_+^a(V_a^-, \psi_-)-\eta_-^a(V_a^+, \psi_+).\label{IV2}
\end{eqnarray}

Due to the detailed analysis of former sections, we only discuss the localization of this model very briefly.
\begin{itemize}
\item Integrate out $L$. This restricts the bosonic fields $x^i$ to the zero locus $N$ of $\sigma$.
\item Integrate out $\chi_\pm$. At the same time, non-zero modes of $\psi_\pm$ transverse to $N$ are integrated out and only the zero modes (tangent to $N$) remain.
\item Integrate out $U$ and non-zero modes of $F$ transverse to $N$. This will give rise to Model II with target space being the invariant submanifold $N$.
\item Follow the localization procedures in Sect.~\ref{LII}. This will finally lead to Model I with $N/G$ being the target space.
\end{itemize}

To conclude this section, we remark that in some cases, the existence of the local potentials $B_\alpha$ is not obvious and we can take the component form (\ref{IV2}) of $S_{IV}$ as the starting point. This is the viewpoint we shall take when we come to the balanced topological Yang-Mills theory in Sect.~\ref{GKR}.
\section{Bismut connections on reduced manifolds }\label{Bismut}
In this section, we describe the metric reduction in terms of purely geometric terms, without referring to any physical ideas. The investigation goes in the spirit of \cite{Ca}, but the role played by Bismut connections is emphasized.
\subsection{Basics of generalized Riemannian geometry}
\label{Bcur}
In this subsection, we recall the most relevant aspects of generalized Riemannian geometry. Though we will finally use the equivalent classical description, this is still a motivating starting point. For a detailed account of notions mentioned below, we refer to \cite{Gua1}.

In contrast to ordinary geometry, in generalized geometry, one considers geometric structures defined on the generalized tangent bundle $\mathbb{T}M=TM\oplus T^*M$, or more generally on an exact Courant algebroid $E$ over a smooth manifold $M$. When referring to integrability of a generalized structure, one often uses the Courant bracket $[\cdot,\cdot]_c$ to replace the Lie bracket on $TM$. Besides, $E$ is also equipped with a non-degenerate pairing $\langle\cdot,\cdot\rangle$ and an anchor map $\pi:E\rightarrow TM$.

Given $E$, one can always find an isotropic splitting $s:TM\rightarrow E$, which has a curvature form $H\in \Omega_{cl}^3(M)$ defined by
\[H(X,Y,Z)=\langle[s(X),s(Y)]_c,s(Z)\rangle,\quad X, Y, Z\in \Gamma(TM).\]
There are many different isotropic splittings, but the relevant curvatures lie in the same cohomology class. By the bundle isomorphism $s+\pi^*:TM\oplus T^*M\rightarrow E$, the Courant algebroid structure of $E$ can be transported onto $\mathbb{T}M$. Then the pairing $\langle\cdot,\cdot\rangle$ is the natural one, i.e.
$\langle X+\xi,Y+\eta\rangle=\xi(Y)+\eta(X)$, and the Courant bracket is
\begin{equation}[X+\xi, Y+\eta]_H=[X,Y]+\mathcal{L}_X\eta-\iota_Yd\xi+\iota_Y\iota_XH,\end{equation}
called the $H$-twisted Courant bracket. Different splittings are related by B-field transforms, i.e. $e^B(X+\xi)=X+\xi+\iota_XB$, where $B$ is a 2-form.

A generalized (Riemannian) metric on $E$ is an orthogonal, self-adjoint endmorphism $\mathcal{G}:E\rightarrow E$ such that $\langle \mathcal{G}e,e\rangle > 0$ for nonzero $e\in E$. It is necessary that $\mathcal{G}^2=id$. The $\pm$-eigenbundles $V_\pm\subset E$ are positive and negative subbudles of maximal rank respectively. A generalized metric induces a natural splitting $E=\mathcal{G}(T^*M)\oplus T^*M$. This is called the \emph{metric splitting}.

Given a generalized metric, we shall always choose the metric splitting. Then $E$ is identified with $\mathbb{T}M$, $\mathcal{G}$ is of the form $\left(\begin{array}{cc} 0 & g^{-1} \\g & 0 \\
\end{array} \right)$ where $g$ is an ordinary Riemannian metric, and vectors in $V_\pm$ are of the form $X\pm g(X)$ respectively for $X\in TM$.

 Denote the curvature of the metric splitting by $H$. Then one can define the $\pm$-Bismut connections $\nabla^\pm=\nabla\pm\frac{1}{2}g^{-1}H$, which play a central role in generalized K$\ddot{a}$hler geometry. It was observed in \cite{Hi} \cite{Gu1} that these connections can be expressed using $H$-twisted Courant bracket:
\begin{equation}[X\mp g(X), Y\pm g(Y)]_H^\pm=\nabla_X^\pm Y\pm g(\nabla_X^\pm Y),\label{Hif}\end{equation}
where $(X+\xi)^\pm$ denote the $V_\pm$-part of $X+\xi\in \mathbb{T}M$ w.r.t. the decomposition $E=V_+\oplus V_-$.

\subsection{Bismut connections on the quotient space}
Let $(M,g, H)$ be a generalized Riemannian manifold (we assume a generalized metric is given and the metric splitting is used to identify $E$ with $\mathbb{T}M$ ). We also assume that a compact, connected Lie group $G$ acts freely and properly on $M$ on the left such that both $g$ and $H$ are invariant.

If no flux $H$ is presented, the Riemannian metric on the quotient space $M^{red}:=M/G$ and its associated Levi-Civita connection can be easily described: A connection of the principal $G$-bundle $M\rightarrow M^{red}$ naturally arises from the $G$-invariant metric, i.e. the horizontal distribution is just the orthogonal complement $\mathcal{H}$ of the vertical distribution. The Levi-Civita connection on $M^{red}$ can then be expressed using the Levi-Civita connection on $M$ and the orthogonal projection from $TM$ to $\mathcal{H}$. But if there is a non-trivial NS-flux $H$ on $M$, the natural connections should be the two Bismut connections $\nabla^\pm$. We address the problem of how to obtain the Bismut connections on $M^{red}$ from that on $M$. This is not as directly derived as in the ordinary case and should be motivated by considerations in generalized geometry. Our approach is based on Eq.~(\ref{Hif}) that the Bismut connections can be expressed using Courant bracket: Since by the reduction procedure established in \cite{BCG1}, the Courant algebroid $E^{red}$ on $M^{red}$ can naturally be described in terms of the Courant algebroid $E$ on $M$, one can expect that the Bismut connections on $M^{red}$ could be described in terms of the Courant bracket on $M$.

 Now assume that the action of $\mathfrak{g}$ is extended by the equivariant 1-form $\xi_{(\cdot)}$ such that $H+\phi^a \xi_a$ is equivariantly closed. Let $K$ be the subbundle of $\mathbb{T}M$ generated by $V_a+\xi_a$, and $K^\bot$ be the orthogonal complement in $\mathbb{T}M$ w.r.t. the natural pairing. Let $K^\mathcal{G}$ be the $\mathcal{G}$-orthogonal complement of $K$ in $K^\bot$. Then it was proved in \cite{BCG1} that $E^{red}:=(K^\bot/K)/G$ acquires a structure of exact Courant algebroid derived from $\mathbb{T}M$. The Courant bracket of two sections $A, B\in \Gamma(E^{red})$ is defined using the Courant bracket $[\hat{A}, \hat{B}]_H$ of their (locally) invariant lifts $\hat{A}, \hat{B}$ in $\mathbb{T}M$.

 \emph{There is a natural isomorphism between $K^\mathcal{G}/G$ and $E^{red}$ defined by projection}. The generalized metric on $E^{red}$ is actually the restriction of $\mathcal{G}$ on the subbundle $K^\mathcal{G}\subset K^\bot$. Accordingly, we have the decomposition $K^\mathcal{G}=V_+^{red}\oplus V_-^{red}$, where $V_\pm^{red}=V_\pm \cap K^\mathcal{G}$.
Project $V_\pm^{red}$ to $TM$. Two horizontal distributions on $M$ arise:
\[\tau_\pm:=\{Y\in TM|g(Y, V_a)\pm \xi_a(Y)=0\}.\]
These are just distributions derived in Eq.~(\ref{con}), and precisely the zero-modes of $\psi_\pm$ in the model there. The advantage of identifying $K^\bot/K$ with $K^\mathcal{G}$ is that, when a lift $\hat{A}\in \Gamma(K^\bot)$ of $A\in \Gamma(E^{red})$ is needed, we can choose $\hat{A}$ to be the unique one in $\Gamma(K^\mathcal{G})$.

Let $\tilde{g}$ be the reduced metric on $M^{red}$. Note that $\tilde{g}$ is by definition derived from restricting $\mathcal{G}$ on $V_+^{red}$ (or $V_-^{red}$). This means $\tilde{g}$ is in fact defined by restricting $g$ on $\tau_+$ (or $\tau_-$). This is different from the ordinary case. Let $\tilde{\nabla}$ be the $-$-Bismut connection on $M^{red}$. Then according to Eq.~(\ref{Hif}), in the metric splitting of $E^{red}$,
\[\tilde{\nabla}_{[X]}[Y]-\tilde{g}(\tilde{\nabla}_{[X]}[Y])=[[X]+\tilde{g}([X]), [Y]-\tilde{g}([Y])]_{\tilde{H}}^-,\]
where $[X]$ is a vector field on $M^{red}$ represented by an invariant lift $X$ on $M$.
But the R.H.S. of the above equation can be computed using the corresponding invariant sections of $K^\mathcal{G}$ (this possibility is explained in detail in \cite{Wang}), i.e.
\[[X^++g(X^+), Y^--g(Y^-)]_H^-,\]
where $X^{\pm}$ denote the horizontal lifts of $[X]$ in $\tau_\pm$ respectively. One should note that $\Gamma^G(K^\mathcal{G})$ is not involutive under the Courant bracket. This can hold only up to addition of invariant section of $K$. Therefore,
\[[X^++g(X^+), Y^--g(Y^-)]_H=A_++A_-+N,\]
where $A_\pm\in V_\pm^{red}$ and $N=2c^a(V_a+\xi_a)$ for some functions $c^a$ to be determined. Of course we want to separate $A_-$ from the above expression. We already have
\[[X^++g(X^+), Y^--g(Y^-)]_H^-=A_-+N_-,\]
where $N_-=c^a(V_a-g(V_a)+\xi_a-g^{-1}\xi_a)$.
Hence,
\[A_-+N_-=\nabla_{X^+}^-Y^--g(\nabla_{X^+}^-Y^-).\]
Therefore,
\[\pi_-(A_-)+c^a(V_a-g^{-1}\xi_a)=\pi_-(A_-)+c^aV^-_a=\nabla_{X^+}^-Y^-,\]
where $\pi_-$ is the projection from $V_-$ to $TM$.
Note that we have the orthogonal decomposition
\[TM=\tau_-\oplus \textup{span}\{V^-_a\}.\]
Thus $\pi_-(A_-)$ is in fact the $\tau_-$-part of $\nabla_{X^+}^-Y^-$ w.r.t. this decomposition. We then find
\[c^a=-T^{ab}(Y^-, \nabla_{X^+}^-V^-_b),\]
where $T_{ab}=g(V^-_a, V^-_b)=g(V_a^+,V_b^+)$ and $T^{ab}$ is its inverse. We finally obtain
\begin{equation}\pi_-(A_-)=\nabla_{X^+}^-Y^-+T^{ab}(Y^-, \nabla_{X^+}^-V^-_b)V^-_a.\label{Bis}\end{equation}
This is what we need to express $\tilde{\nabla}$ in terms of $\nabla^-$; in particular, if $[Z]$ is another vector field on $M^{red}$, then
\begin{equation}(\tilde{\nabla}_{[X]}[Y],[Z])=(\nabla_{X^+}^-Y^-+T^{ab}(Y^-, \nabla_{X^+}^-V^-_b)V^-_a,Z^-)=(\nabla_{X^+}^-Y^-, Z^-).\label{rBis}\end{equation}

Now we can turn to the problem of expressing the curvature of $\tilde{\nabla}$ in terms of that of $\nabla^-$. From Eq.~(\ref{Bis}) and Eq.~(\ref{rBis}), we have (for the detail, see \cite{Wang})
\[(\tilde{\nabla}_{[X]}\tilde{\nabla}_{[Y]}[Z],[W])=(\nabla_{X^+}^-\nabla_{Y^+}^-Z^-, W^-)+T^{ab}(Z^-, \nabla_{Y^+}^-V^-_a)(W^-, \nabla_{X^+}^-V^-_b).\]
Again due to Eq.~(\ref{rBis}),
\begin{eqnarray*}(\tilde{\nabla}_{[[X],[Y]]}[Z], [W])&=&(\nabla^-_{[X^+, Y^+]+\Omega_+(X^+, Y^+)}Z^-, W^-)\\
&=&(\nabla^-_{[X^+, Y^+]}Z^-, W^-)+(\nabla^-_{\Omega_+(X^+, Y^+)}Z^-, W^-),\end{eqnarray*}
where we have used the identity\footnote{$\widetilde{[X]}$ denotes the horizontal lift of $[X]$ in $\tau_+$.}
\[[X^+, Y^+]-\widetilde{[[X], [Y]]}=-\Omega_+(X^+, Y^+)=-\Omega^a_+(X^+, Y^+)V_a,\]
and $\Omega^a_+$ is the curvature associated to $\tau_+$. We want to use the data $V_a$, $\xi_a$ to represent $\Omega_+^a$. Let $\theta_+$ be the connection 1-form associated with $\tau_+$. Then
\[\theta_+=\theta_+^ae_a=t^{ba}g(V_b^+)e_a,\]
where $t^{ba}$ is to be determined. We have
\[t^{ba}(V_b^+,V_c)=t^{ba}K_{cb}=\delta^a_c,\]
where $K_{ab}$ is defined in Sect.~\ref{LII}. Then $t^{ba}$ is precisely $K^{ba}$ in Sect.~\ref{LII} and $\theta_+^a=K^{ba}g(V_b^+)$.
\begin{lemma}Let $\Omega_\pm^a$ be the curvatures of $\tau_\pm$. Then
\[\Omega_+^a|_{\tau_+}=K^{ba}d\xi^+_b|_{\tau_+},\quad \Omega_-^a|_{\tau_-}=K^{ab}d\xi^-_b|_{\tau_-},\]
where $\xi^\pm_b=g(V_b^\pm)$.
\end{lemma}
\begin{proof}We only compute $\Omega_+^a$. The computation for $\Omega_-^a$ is similar.
\begin{eqnarray*}
\Omega_+^a(X^+, Y^+)&=&d\theta_+^a(X^+, Y^+)=X^+\theta_+^a(Y^+)-Y^+\theta_+^a(X^+)-\theta_+^a([X^+, Y^+])\\
&=&-\theta_+^a([X^+, Y^+])=-K^{ba}(V_b^+,[X^+, Y^+])\\
&=&K^{ba}(d\xi^+_b)(X^+, Y^+).
\end{eqnarray*}
\end{proof}
We then have
\[(\nabla_{\Omega_+(X^+, Y^+)}^-Z^-, W^-)=K^{ba}(d\xi^+_b)(X^+, Y^+)(\nabla_{V_a}^-Z^-, W^-).\]
Note that \begin{eqnarray*}(\nabla_{V_a}^-Z^-, W^-)&=&(\nabla_{V_a}Z^-,W^-)-\frac{1}{2}H(V_a, Z^-, W^-)\\
&=&(\nabla_{V_a}Z^-,W^-)-\frac{1}{2}(d\xi_a)(Z^-, W^-),\end{eqnarray*}
\begin{eqnarray*}(\nabla_{V_a}Z^-,W^-)&=&(\nabla_{Z^-}V_a, W^-)=Z^-(g(V_a)(W^-))-(V_a, \nabla_{Z^-}W^-)\\
&=&Z^-(g(V_a)(W^-))-(V_a, \nabla_{W^-}Z^-)-g(V_a)([Z^-, W^-])\\
&=&Z^-(g(V_a)(W^-))-W^-(g(V_a)(Z^-))+(\nabla_{W^-}V_a, Z^-)\\&-&g(V_a)([Z^-, W^-])\\
&=&dg(V_a)(Z^-, W^-)+(\nabla_{W^-}V_a, Z^-),
\end{eqnarray*}
and
\[(\nabla_{V_a}Z^-,W^-)+(\nabla_{W^-}V_a, Z^-)=0.\]
Then we have
\begin{equation}(\nabla_{V_a}^-Z^-, W^-)=\frac{1}{2}d\xi^-_a(Z^-, W^-).\label{3form}\end{equation}

Now we can finally find the curvature $\tilde{R}$ of $\tilde{\nabla}$ in terms of $R^-$.
\begin{theorem}The curvature $\tilde{R}$ of $\tilde{\nabla}$ is
\begin{eqnarray*}(\tilde{R}([X], [Y])[Z], [W])&=&(R^-(X^+, Y^+)Z^-, W^-)\\
&-&\frac{K^{ab}}{2}(d\xi^+_a)(X^+, Y^+)(d\xi^-_b)(Z^-, W^-)\\&+&T^{ab}[(Z^-, \nabla_{Y^+}^-V^-_a)(W^-, \nabla_{X^+}^-V^-_b)-(X\leftrightarrow Y)].
\end{eqnarray*}
\end{theorem}
To see this coincides with the computation in Sect.~\ref{LII}, note that for $\nabla$ we have $d\xi^\pm_a=\nabla\xi^\pm_a$ and that
\begin{eqnarray*}(Z^-, \nabla_{Y^+}^-V^-_a)&=&(Z^-, \nabla_{Y^+}V^-_a)-\frac{1}{2}H(Y^+,V^-_a,Z^-)\\
&=&\frac{1}{2}(Z^-, \nabla_{Y^+}V^-_a)+\frac{1}{2}(Z^-, \nabla_{Y^+}V_a)-\frac{1}{2}(Z^-, \nabla_{Y^+}g^{-1}\xi_a)\\
&+&\frac{1}{2}(\nabla \xi_a)(Y^+, Z^-)-\frac{1}{2}H(g^{-1}\xi_a, Y^+,Z^-)\\
&=&\frac{1}{2}(Z^-, \nabla_{Y^+}V^-_a)-\frac{1}{2}(\nabla_{Z^-}V_a, Y^+)-\frac{1}{2}(Y^+, \nabla_{Z^-}g^{-1}\xi_a)\\
&-&\frac{1}{2}H(g^{-1}\xi_a, Y^+,Z^-)\\
&=&\frac{1}{2}(Z^-, \nabla_{Y^+}V^-_a)-\frac{1}{2}(\nabla_{Z^-}V_a^+, Y^+)-\frac{1}{2}H(g^{-1}\xi_a, Y^+,Z^-).\end{eqnarray*}
These are enough to identify the above curvature with Eq.~(\ref{CurII}). The above computation also reveals that Eq.~(\ref{Bis}) coincides with Eq.~(\ref{bis}).

In the remainder of this subsection, we compute the reduced NS-flux $\tilde{H}$ on $M^{red}$. The appearance of the final expression depends on which distribution is used to model $TM^{red}$. There are several natural choices (at least three, namely $\tau_\pm$ and $\tau$ in \cite[Prop.~4.2]{Ca}) to achieve this purpose. However, in the literature, such as \cite{Ca} and \cite{Hi2}, $\tau_+$ was often used for this purpose and we shall follow this convention to compare different viewpoints towards this computation.

We have an analogue of Eq.~(\ref{rBis}) to express the $+$-Bismut connection $\hat{\nabla}$ on $M^{red}$, i.e.
\[(\hat{\nabla}_{[X]}[Y],[Z])=(\nabla_{X^-}^+Y^+, Z^+).\]
Note that $\tilde{H}$ is the torsion of $\hat{\nabla}$, namely
\[\tilde{H}([X],[Y],[Z])=(\hat{\nabla}_{[X]}[Y],[Z])-(\hat{\nabla}_{[Y]}[X],[Z])-([[X],[Y]],[Z]).\]
It is elementary to find
\[X^-=X^++2K^{ab}\xi_b(X^+)V_a,\]
and an analogue of Eq.~(\ref{3form}):
\[(\nabla_{V_a}^+Y^+, Z^+)=\frac{1}{2}d\xi^+_a(Y^+, Z^+).\]
Besides, we have
\begin{eqnarray*}([[X],[Y]],[Z])&=&([X^+,Y^+]+\Omega_+^a(X^+,Y^+)V_a, Z^+)\\
&=&([X^+,Y^+],Z^+)-\Omega_+^a(X^+,Y^+)\xi_a(Z^+),\end{eqnarray*}
where we have used $(V_a^+, Z^+)=0$.

Piecing these formulae together, we have
\begin{eqnarray*}\tilde{H}([X],[Y],[Z])&=&H(X^+,Y^+,Z^+)+K^{ab}\xi_b(X^+)d\xi_a^+(Y^+,Z^+)\\&-&K^{ab}\xi_b(Y^+)d\xi_a^+(X^+,Z^+)
+\Omega_+^a(X^+,Y^+)\xi_a(Z^+)\\
&=&H(X^+,Y^+,Z^+)+\Omega_+^a(Y^+,Z^+)\xi_a(X^+)\\&-&\Omega_+^a(X^+,Z^+)\xi_a(Y^+)
+\Omega_+^a(X^+,Y^+)\xi_a(Z^+)\\
&=&(H+\Omega_+^a\wedge\xi_a)(X^+,Y^+,Z^+),
\end{eqnarray*}
where we have used the formula $\Omega^a_+(X^+, Y^+)=K^{ba}(d\xi^+_b)(X^+, Y^+)$ derived before. The result coincides with the computation in \cite{Ca}.
\subsection{Bismut connections on the submanifold}
In this subsection, we express the curvature of the $-$-Bismut connection on $N=\sigma^{-1}(0)$ in Sect.~\ref{Model3} in terms of the curvature of $\nabla^-$ on $M$. The discussion is along the same line of the former subsection, but is much easier to carry out. In this cotangent action case, the Courant algebroid over $M$ descends rather directly to $N$. In present setting, $K|_N$ is the co-normal bundle generated locally by $\{d\sigma^\alpha\}$. Hence $K^\bot=TN\oplus T^*M|_N$ and
\[K^\mathcal{G}=\{X+\xi|X\in TN,\quad \xi\in T^*M|_N, \quad g(\xi, d\sigma^\alpha)|_N=0\}\]
can still be identified with $K^\bot/K|_N$. Note that as an exact Courant algebroid over $N$, $K^\mathcal{G}$ is already in the metric splitting and the corresponding curvature is $H|_N$, where $H$ is the curvature of the metric splitting of $E$ over $M$.

On $N$, the $-$-Bismut connection must satisfy
\[\tilde{\nabla}_{\bar{X}}\bar{Y}-g(\tilde{\nabla}_{\bar{X}}\bar{Y})=[\bar{X}+g(\bar{X}), \bar{Y}-g(\bar{Y})]_{H|_N}^-,\]
where $\bar{X}$, $\bar{Y}$ are tangent vector fields over $N$. The R.H.S. of the above equation can be computed using its extension in $\Gamma(K^\mathcal{G})$:
\[[X+g(X), Y-g(Y)]_H^-,\]
where $X$, $Y$ are (local) extensions of $\bar{X}$, $\bar{Y}$ respectively. But $K^\mathcal{G}$ is involutive, under the Courant bracket, only up to addition of a section of $K$. Therefore,
\[[X+g(X), Y-g(Y)]_H=A_++A_-+P\quad on\quad N,\]
where $A_\pm\in \Gamma(V_\pm\cap K^\bot)$, and $P=2c_\alpha d\sigma^\alpha|_N$ for $c_\alpha$ to be determined. We shall separate $A_-$ from the above expression. We already have
\[[X+g(X), Y-g(Y)]_H^-=A_-+P_-=\nabla_X^-Y-g(\nabla_X^-Y)\quad on\quad N,\]
where
\[P_-=c_\alpha(-g^{-1}d\sigma^\alpha+d\sigma^\alpha)|_N.\]
Therefore,
\[\pi_-(A_-)-c_\alpha g^{-1}d\sigma^\alpha=\nabla_X^-Y,\quad on\quad N\]
where $\pi_-: V_-\rightarrow TM$ is the projection,
and
\[(\pi_-(A_-), g^{-1}d\sigma^\beta)|_N-c_\alpha(g^{-1}d\sigma^\alpha, g^{-1}d\sigma^\beta)|_N=(\nabla_X^-Y, g^{-1}d\sigma^\beta)|_N,\]
i.e.
\[-c_\alpha T^{\alpha\beta}=d\sigma^\beta(\nabla_X^-Y)|_N,\]
where $T^{\alpha\beta}$ is defined as in Sect.~\ref{Model3}: $T^{\alpha\beta}=g(d\sigma^\alpha, d\sigma^\beta)|_N$.

We thus find
\[c_\alpha=T_{\alpha\beta}(Y, \nabla_X^- d\sigma^\beta)|_N.\]
Therefore,
\begin{equation}\pi_-(A_-)|_N=\nabla_X^-Y|_N+T_{\alpha\beta}(Y, \nabla_X^- d\sigma^\beta)|_N(g^{-1}d\sigma^\alpha)|_N.\label{sbis2}\end{equation}
This is what we need to express $\tilde{\nabla}$ on $N$ in terms of $\nabla^-$ on $M$; in particular, if $\bar{Z}$ is another vector field over $N$, then
\[(\tilde{\nabla}_{\bar{X}}\bar{Y}, \bar{Z})=(\nabla_X^-Y+T_{\alpha\beta}(Y, \nabla_X^- d\sigma^\beta)g^{-1}d\sigma^\alpha, Z)|_N=(\nabla_X^-Y, Z)|_N.\]
It is not hard to obtain
\[(\tilde{\nabla}_{\bar{X}}\tilde{\nabla}_{\bar{Y}}\bar{Z},\bar{W})=(\nabla_X^-\nabla_Y^- Z, W)|_N+T_{\alpha\beta}(Z, \nabla_Y^-d\sigma^\beta) (W,\nabla_X^-d\sigma^\alpha)|_N.\]
We also have the simple identity
\[(\tilde{\nabla}_{[\bar{X},\bar{Y}]}\bar{Z}, \bar{W})=(\nabla_{[X, Y]}^-Z, W)|_N.\]
Combining all we have got above together, we obtain
\begin{theorem}The curvature $\tilde{R}$ of $\tilde{\nabla}$ is
\begin{eqnarray*}(\tilde{R}(\bar{X}, \bar{Y})\bar{Z}, \bar{W})&=&(R^-(X, Y)Z, W)|_N\\&+&T_{\alpha\beta}[(Z, \nabla_Y^- d\sigma^\beta)(W, \nabla_X^- d\sigma^\alpha)-(Z, \nabla_X^- d\sigma^\beta)(W, \nabla_Y^- d\sigma^\alpha)]|_N.\end{eqnarray*}
\end{theorem}
To see this coincides with Eq.~(\ref{CurIII}), one only needs to note that
\begin{eqnarray*}(Z, \nabla_Y^- d\sigma^\beta)|_N&=&(Z, \nabla_Y d\sigma^\beta)|_N-\frac{1}{2}H(Y, g^{-1}d\sigma^\beta, Z)|_N\\
&=& (\nabla d\sigma^\beta)(Y,Z)+(Y, \nabla_Z d\sigma^\beta)|_N+\frac{1}{2}H(Z, g^{-1}d\sigma^\beta, Y)|_N\\
&=&(Y, \nabla_Z d\sigma^\beta)|_N+\frac{1}{2}H(Z, g^{-1}d\sigma^\beta, Y)|_N\\
&=&(Y,\nabla_Z^+ d\sigma^\beta)|_N,
\end{eqnarray*}
where we have used the fact that $\nabla d\sigma^\beta=d^2\sigma^\beta=0$. The above computation also reveals that Eq.~(\ref{sbis2}) really coincides with Eq.~(\ref{sbis1}).
\section{Generalized K$\ddot{a}$hler reduction}
\label{GKR}
At first glance, to encode generalized K$\ddot{a}$hler reduction in our Model IV, one should extend N=2 supersymmetry to N=4, just as what physicists do in 2-dimensional supersymmetric $\sigma$-models. However, it is not the case here. N=4 supersymmetry is a too strong constraint to impose. In the reduced model, the field content is also reduced--the contribution of $W$ and $\sigma$ is just to single out the submanifold $\sigma^{-1}(0)$, and after that, no freedom from the fibers of $W$ remains in the reduced model. What we really need is to make the reduced model, rather than the original model, to have N=4 supersymmetry.

On the mathematical side, in terms of generalized geometry, a generalized K$\ddot{a}$hler structure is a generalized Riemannian manifold equipped with a compatible generalized complex structure $\mathcal{J}_1$ such that both $\mathcal{J}_1$ and $\mathcal{G}\mathcal{J}_1$ as almost generalized complex structures are integrable. In the case of a generalized K$\ddot{a}$hler manifold $M$ carrying a trivially extended action $V_a+\xi_a$ (we assume that this infinitesimal action preserves the generalized K$\ddot{a}$hler structure and can be integrated to a group action) and a moment map $\mu$ (an equivariant map from $M$ to the dual space $V^*$ of a $\mathfrak{g}$-module $V$) with $0$ as a regular value, one first singles out the submanifold $N=\mu^{-1}(0)$ and forms the bundle $K$ over $N$, locally generated by $V_a+\xi_a$ and $d\mu^\alpha$. Then $K^\bot$ is again defined as the orthogonal complement of $K$ in $\mathbb{T}M|_N$ and one gets the important bundle $K^\mathcal{G}=K^\bot \cap \mathcal{G}(K^\bot)$ over $N$. A sufficient condition for the generalized K$\ddot{a}$hler structure to descend to $N/G$ is that $\mathcal{J}_1$ preserves $K^\mathcal{G}$, i.e.,
\begin{equation}\mathcal{J}_1K^\mathcal{G}=K^\mathcal{G}.\label{kcon}\end{equation}
One can carry out the metric reduction first and find that Eq.~(\ref{kcon}) ensures that the reduced Courant algebroid acquire an almost generalized K$\ddot{a}$hler structure. As for integrability of this structure, it stems from the integrability of a general reduced Dirac structure in the context of Courant reduction. Much more details of this generalized geometric approach to generalized K$\ddot{a}$hler reduction can be found in \cite{BCG1} \cite{Ca}.

Let us explain the reduction procedure in some detail from another angle. In terms of more familiar ordinary geometric notions, $K$ and $\mathcal{G}$ determine two horizontal distributions $\tau_\pm$ over $N$ viewed as a principal $G$-bundle. Then Eq.~(\ref{kcon}) is just that $J_\pm$ preserves the distributions $\tau_\pm$ respectively, i.e.
\begin{equation}J_\pm \tau_\pm=\tau_\pm,\label{kc}\end{equation} where $J_\pm$ are the two complex structures underlying the biHermitian description of the generalized K$\ddot{a}$hler structure. As $T(N/G)$ is modeled on $\tau_\pm$, Eq.~(\ref{kc}) implies that $N/G$ acquires two almost complex structures compatible with the reduced metric $\tilde{g}$. This is the viewpoint of \cite{Ca} towards generalized K$\ddot{a}$hler reduction.

To see $N/G$ is really generalized K$\ddot{a}$hlerian in the spirit of \cite{Ca} needs more effort: Let $\tilde{J}_\pm$ be the reduced almost complex structures on $N/G$. Firstly, one should prove that $\tilde{J}_\pm$ are flat w.r.t. the reduced $\pm$-Bismut connections respectively. This can be easily achieved by using the formulae relating the Bismut connections in $M$ and in $N/G$. Secondly, one should prove that the reduced NS-flux $\tilde{H}$ is of type $(1,2)+(2,1)$ w.r.t. both $\tilde{J}_\pm$. For $\tilde{J}_+$, from the formula $\tilde{H}=(H+\Omega^a_+\wedge \xi_a)|_{\tau_+}$, it's enough to prove that the curvature $\Omega_+$ is of type $(1,1)$ w.r.t. $\tilde{J}_+$. The details of this computation and further discussions can be found in \cite{Wang}. The conclusion for $\tilde{J}_-$ holds similarly.

Now we come back to our Model IV and see briefly how the above generalized K$\ddot{a}$hler reduction is realized physically. Since we are only concerned with freedoms in the reduced space $\sigma^{-1}(0)/G$, we won't bother ourselves to consider the extended supersymmetric transform of $\zeta$ and $\phi$. As in the non-gauged model, the extended supersymmetric transform of $\varphi^i$ is
\[\delta'_\epsilon \varphi^i=\epsilon^+ J_{+j}^i\mathcal{Q}_+\varphi^j+\epsilon^-J_{-j}^i\mathcal{Q}_-\varphi^j.\]In components, this supersymmetric transform of $\varphi^i$ is
\[d_\pm'x^i=J_{\pm j}^i \psi_\pm^j,\quad d_{\pm}'\psi_\pm^i=-J_{\pm j}^iV_a^j \phi_{\pm\pm}^a-J_{\pm j,k}^i\psi^k_\pm\psi^j_\pm,\]
\[d_\pm' \psi_\mp^i=-J_{\pm j}^i(\phi_{+-}^aV_a^j\pm \tilde{F}^j)-J_{\pm j,k}^i\psi_\mp^k\psi_\pm^j,\]
\begin{eqnarray*}d_\pm' \tilde{F}^i&=&J_{\pm j}^i(2\eta_{\pm}^aV_a^j\pm \phi_{\pm\pm}^aV_{a,k}^j\psi_\mp^k\mp \phi_{+-}^aV_{a,k}^j\psi_\pm^k)\\&+&J_{\pm j,k}^i(\tilde{F}^k\psi_\pm^j-\tilde{F}^j\psi_\pm^k\mp\phi_{+-}^aV_a^j\psi_\pm^k\pm\phi_{\pm\pm}^aV_a^j\psi_\mp^k).\end{eqnarray*}

In the reduced model, one should only pay attention to bosonic freedoms transverse to $G$-orbits in $\sigma^{-1}(0)$ and fermionic freedoms satisfying
\[(V^\pm_a, \psi_\pm)=0,\quad \partial_i\sigma^\alpha\psi_\pm^i=0.\]
To make such $(x^i, \psi_\pm^i)$ form an on-shell supermultiplet of N=4 algebra, $J_\pm \tau_\pm=\tau_\pm$ is a natural constraint to impose on the zero modes of $\psi_\pm$ and our former discussion really assures that the reduced complex structures $\tilde{J}_\pm$ is enough to realize on-shell N=4 supersymmetry in the reduced model.
\section{The moduli space of instantons}\label{MOI}
In this section, as a non-trivial application of the formalism developed in former sections, we account for the generalized K$\ddot{a}$hler structure on the moduli space of instantons investigated in \cite {Hi2} \cite{BCG2}. As in current literature the origin of this structure is mathematically explained very clearly, we content ourselves with writing down the action.

Let $G$ be compact, connected and semi-simple, and $P\rightarrow M$ be a principal $G$-bundle over a smooth oriented generalized Riemannian 4-manifold $(M, g,H)$. Let $\mathfrak{g}_P$ be the adjoint bundle associated to $P$. The space $\mathcal{A}$ of all connections on $P$ is an affine space modeled on $\Omega^1(\mathfrak{g}_P)$. The gauge group $\mathfrak{G}$ acts on $\mathcal{A}$ and has $\Gamma(\mathfrak{g}_P)$ as its Lie algebra. Note that on $\Omega^\bullet(\mathfrak{g}_P)$, we have the natural gauge invariant metric
\[(\alpha, \beta)=\int_M \kappa(\alpha, \ast \beta),\]
where $\kappa$ is the metric on $\mathfrak{g}$ induced from the Killing form of $\mathfrak{g}$ and $\ast$ is the Hodge star operator associated to $g$. The moduli space $\mathcal{M}$ of instantons is obtained from $\mathcal{A}$ by first imposing the anti-self-dual (ASD) equation $F_A^+=0$ and then quotienting by $\mathfrak{G}$\footnote{In fact, to have a smooth structure on $\mathcal{M}$, one should start from the subspace $\mathcal{A}^*$ of irreducible connections rather than the whole of $\mathcal{A}$.}.

For $\gamma \in \textup{Lie}(\mathfrak{G})$, the vector field generated by $\gamma$ on $\mathcal{A}$ is $D_A\gamma$ at $A\in \mathcal{A}$, where $D_A$ is the exterior covariant derivative w.r.t. $A$. The 1-form $\xi$ generated by $\gamma$ is $-H\gamma$ (lying in $\Omega^3(\mathfrak{g}_P)$ which, via the metric, can be viewed as the cotangent space at $A$). Now $\mathcal{A}$ can be viewed as an infinite-dimensional generalized Riemannian manifold with vanishing NS-flux. With this understanding, the formula $d\xi_a=\iota_aH$ naturally holds because $\mathcal{A}$ is flat and $\xi$ is translation-invariant \cite{BCG2}.

Now it is totally clear how to formally encode the reduction procedure using our Model IV. Here we directly use the action (\ref{IV2}) instead of (\ref{IV1}) because it seems that no obvious potential $B$ exists in this setting. Now $\psi_\pm$, $F$ are 1-forms living in $\Omega^1(\mathfrak{g}_P)$ and $\phi$, $\eta$ are elements in $\Gamma(\mathfrak{g}_P)$. The vector bundle $W$ is a trivial one with the self-dual part $\Omega^2_+(\mathfrak{g}_P)$ of $\Omega^2(\mathfrak{g}_P)$ as the fibre, and $\sigma(A)=F_A^+=(dA+\frac{1}{2}[A,A])^+$, namely, the self-dual part of the curvature associated to $A$. Let us write down the associated action.

\begin{itemize}
\item Since $\mathcal{A}$ is flat and equipped with the zero NS-flux, in the action of $S_{IV}$, there is only one term $\frac{1}{2}(F, F)$ in $S_I$.
\item The terms involving the ASD equation is of the following form:
\[\sqrt{-1}[\int_M\kappa(\sigma, L)+\kappa(D_A^+\psi_-, \chi_+)-\kappa(D_A^+\psi_-, \chi_-)+\kappa([\psi_-, \psi_+]^++D_A^+F, U)],\]
where $D_A^+$ is the self-dual part of $D_A$ acting on $\Omega^1(\mathfrak{g}_P)$ and $[\psi_-, \psi_+]^+$ the self-dual part of $[\psi_-, \psi_+]$.
\item The remaining terms involving $\phi$, $\eta$ are collected as follows:
\begin{eqnarray*}\int_M &\quad&\kappa(H \phi_{+-}, F)-\frac{1}{2}\kappa(\phi_{+-}, \ast\phi_{+-})+\kappa(\psi_+, \ast[\psi_-, \phi_{+-}])\\
&+&\frac{1}{2}\kappa(\phi_{++},\ast \phi_{--})+\frac{1}{2}\kappa([\psi_-, \phi_{++}], \ast \psi_-)+\frac{1}{2}\kappa([\psi_+, \phi_{--}], \ast \psi_+)\\
&-&\kappa(H\phi_{--},D_A\phi_{++})+\kappa(D_A\eta_+, \ast \psi_-)+\kappa(H\eta_+, \psi_-)\\
&-&\kappa(D_A\eta_-, \ast \psi_+)+\kappa(H\eta_-, \psi_+).\end{eqnarray*}
\end{itemize}

If additionally $M$ is equipped with an even generalized K$\ddot{a}$hler structure, then $\mathcal{A}$ is generalized K$\ddot{a}$hlerian with the underlying gauge invariant complex structures $\mathbb{J}_\pm$ defined naturally by $J_\pm$ acting on $\Omega^1(\mathfrak{g}_P)$. That $\tau_\pm$ are invariant under $\mathbb{J}_\pm$ respectively was already contained in \cite[Chap.~5.3]{LT}. Thus by our formalism, $\mathcal{M}$ is generalized K$\ddot{a}$hlerian.

\section{Conclusion}
To finish our present investigation, we point out some problems left for future work.

Firstly, in \cite{DM}, an N=2 cohomology theory and its equivariant version were sketched briefly as the mathematical background of balanced topological field theories. It was argued that the resulting cohomology would not give much more information than the traditional de Rham cohomoloy. However, by rule of thumb, generalized geometry suggests that what really of interest is the twisted de Rham cohomology associated to the twisted differential $d-H\wedge$, where the NS-flux $H$ plays an essential role. In balanced theories, the cohomology investigated in \cite{DM} has nothing to do with this $H$ and $H$ only enters the theory via our action $S_I$ or $S_{II}$. This suggests that in balanced theories there would be another cohomology theory with an action as an essential ingredient. Additionally, the relevant notion of generalized spinors is another important and useful part of generalized geometry. It would be of interest to see how generalized spinors arise in balanced topological field theories.

Secondly, in the literature, such as \cite{VW} \cite{DM} \cite{DPS} \cite{Sa}, there are other balanced topological Yang-Mills actions different from ours. It's interesting to investigate the difference in some detail and even compute the partition function of our model explicitly in some special situations.

 \section*{Acknowledgemencts}
 This study is supported by the Natural Science Foundation of Jiangsu Province (BK20150797).

\end{document}